\documentclass[showpacs,preprintnumbers,amsmath,amssymb,aps,pre,onecolumn]{revtex4}

\usepackage{amsmath}
\usepackage{graphicx}
\usepackage{dcolumn}
\usepackage{bm}
\usepackage{epsfig}

\newtheorem{theorem}{Theorem}[section]
\newenvironment{proof}[1][Proof]{\begin{trivlist}
\item[\hskip \labelsep {\bfseries #1}]}{\end{trivlist}}
\newcommand{\qed}{\nobreak \ifvmode \relax \else
      \ifdim\lastskip<1.5em \hskip-\lastskip
      \hskip1.5em plus0em minus0.5em \fi \nobreak
      \vrule height0.75em width0.5em depth0.25em\fi}
\begin{document}


\title{DETECTING SERIES PERIODICITY WITH HORIZONTAL VISIBILITY GRAPHS}
\author{ANGEL NU\~{N}EZ, LUCAS LACASA, EUSEBIO VALERO,\\ JOSE PATRICIO G\'{O}MEZ, BARTOLO LUQUE}

\address{Departamento de Matem\'{a}tica Aplicada y Estad\'{i}stica\\
ETSI Aeron\'{a}uticos, Universidad Polit\'{e}cnica de Madrid, Spain\\
am.nunez@alumnos.upm.es}

\begin{abstract}
The \textit{horizontal visibility algorithm} has been recently introduced as a mapping between time series and networks. The challenge lies in characterizing the structure of time series (and the processes that generated those series) using the powerful tools of graph theory. Recent works have shown that the visibility graphs inherit several degrees of correlations from their associated series, and therefore such graph theoretical characterization is in principle possible. However, both the mathematical grounding of this promising theory and its applications are on its infancy. Following this line, here we address the question of detecting hidden periodicity in series polluted with a certain amount of noise. We first put forward some generic properties of horizontal visibility graphs which allow us to
define a (graph theoretical) noise reduction filter. Accordingly, we evaluate its performance for the task of calculating the period of noisy periodic signals, and compare our results with standard time domain (autocorrelation) methods. Finally, potentials, limitations and applications are discussed.
\end{abstract}

 \maketitle

\section{Introduction}
In the last years, some methods mapping time series to network representations have been proposed (see for instance [Xu \textit{et al.} 2008, Zhang \& Small 2006, Lacasa \textit{et al.} 2008, Luque \textit{et al.} 2009] and a recent review on this topic [Donner \textit{et al.} 2011]). The general purpose is to investigate on
the properties of the series through graph theoretical tools recently developed in the core of the celebrated complex network theory, opening the possibility of building bridges between time series analysis, nonlinear dynamics, and graph theory. Along this line, the family of visibility algorithms [Lacasa \textit{et al.} 2008, Luque \textit{et al.} 2009, Gutin \textit{et al.} 2011] has been introduced recently. It has been shown that several degrees of correlations (including periodicity [Lacasa \textit{et al.} 2008], fractality [Lacasa \textit{et al.} 2009] or chaoticity [Luque \textit{et al.} 2009, Lacasa \& Toral 2010]) can be captured by the algorithm and translated in the associated visibility graph. Accordingly, several works applying such algorithm in several contexts ranging from geophysics [Elsner \textit{et al.} 2009] or turbulence [Liu \textit{et al.} 2010] to physiology [Shao 2010] or finance [Yang \textit{et al.} 2010] have started to appear. Here we focus on a specific algorithm within this family called the \textit{horizontal} visibility graph [Luque \textit{et al.} 2009], which has been recently considered for the task of discriminating chaotic from correlated stochastic processes [Lacasa \& Toral 2010]. While the first steps for a rigorous mathematical grounding have been reported recently [Gutin \textit{et al.} 2011], this method is currently largely unexplored, both from a purely theoretical or from an applied point of view. To partially solve such issues, in this work we address the task of filtering a noisy signal with a hidden periodic component within the horizontal visibility formalism, that is, we explore the possibility of using the method for noise filtering purposes. Periodicity detection algorithms (see for instance [Parthasarathy \textit{et al.} 2006]) can be classified in essentially two categories, namely the time domain (autocorrelation based) and frequency domain (spectral) methods. Here we make use of the horizontal visibility algorithm to propose a third category, namely graph theoretical methods.\\
The rest of the paper goes as follows: in section 2 we present the method, and provide some theorems regarding several topological properties of horizontal visibility graphs. In section 3 we introduce the concept of a graph-theoretical noise filter, and provide some examples of noisy periodic series, comparing in each case the performance of the proposed method with an autocorrelation function analysis. A pathological case that yields misleading results from the autocorrelation function is also considered. We finally provide a discussion on the the potentials and limitations of this approach.

\section{Horizontal visibility algorithm}
The horizontal visibility algorithm has been recently introduced [Luque \textit{et al.} 2009] as a map between a time series and a
graph and it is defined as follows. Let $\{x_i\}_{i=1,. . .,N}$ be a time
series of $N$ real data. The algorithm assigns each datum of the
series to a node in the {\sl horizontal visibility graph} (HVG). Accordingly, a series of $N$ data map to an HVG with $N$ nodes.
Two nodes $i$ and $j$ in the graph are connected if one
can draw a horizontal line in the time series joining $x_i$ and $x_j$
that does not intersect any intermediate data height. That is, $i$ and $j$ are two connected
nodes if the following geometrical criterion is fulfilled
within the time series:
\begin{equation}
x_i,x_j > x_n, \ \forall \ n \ \left | \ i < n < j\right. . \label{criterio}
\end{equation}
Basic properties of this graphs can be found in [Luque \textit{et al.} 2009], and the first steps for a rigorous mathematical characterization can be found in[Gutin \textit{et al.} 2011]. Among possible applications of the method for time series analysis purposes, discrimination between chaotic and stochastic signals has been recently addressed [Lacasa \& Toral, 2010].\\
In this section we provide some theorems regarding some specific topological properties of the horizontal visibility graphs, and in the following sections we will rely on these theorems to define a noise filtering technique.\\

\begin{theorem}
\label{kmedia}
(Mean degree of periodic series)\\The mean degree of an horizontal visibility graph associated to an infinite periodic series of period T (with no repeated values within a period) is
$$\bar k(T)=4\bigg(1-\frac{1}{2T}\bigg)$$
\end{theorem}
\begin{proof}
Without lack of generality, represent the series as
$\{...,x_0,x_1,...,x_T,x_1,x_2,...\}$, where $x_0=x_T$ corresponds
to the largest value of the series. By construction, the
associated HVG is composed as a concatenation of identical motifs,
each of these motifs being itself an HVG of $T+1$ nodes associated
to the subseries $x_0,x_1,...,x_T$, and the mean degree of the HVG
$\bar k$ corresponds to the mean degree of the motif constructed
with $T$ nodes (the nodes associated to $x_0$ and $x_T$ only
introduce half of their actual degree in the motif, what is
equivalent to effectively reducing one node). Suppose that the
motif is a graph with $V$ edges, and let $x_i$ be the smallest
datum of the subseries (since no repetitions are allowed in the
motif, $x_i$ will always be well defined). By construction, the
associated node $i$ will have degree $k=2$. Extract now from the
motif this node and its two edges. The resulting motif will have
$V-2$ edges and $T$ nodes. Iterate this operation $T-1$ times (see
figure \ref{ejemplo_eli} for a graphical illustration of this
process). The resulting graph will have only two nodes, associated
to $x_0$ and $x_T$, linked by a single edge, and the total number
of deleted edges will be $2(T-1)$. Hence, $$\bar k \equiv
2\frac{\# \ edges}{\# \ nodes}=\frac{2(2(T-1)+1)}{T}\Rightarrow
\bar k=4\bigg(1-\frac{1}{2T}\bigg).$$\qed
\end{proof}
\begin{figure}[h]
\centering
\includegraphics[width=0.85\textwidth]{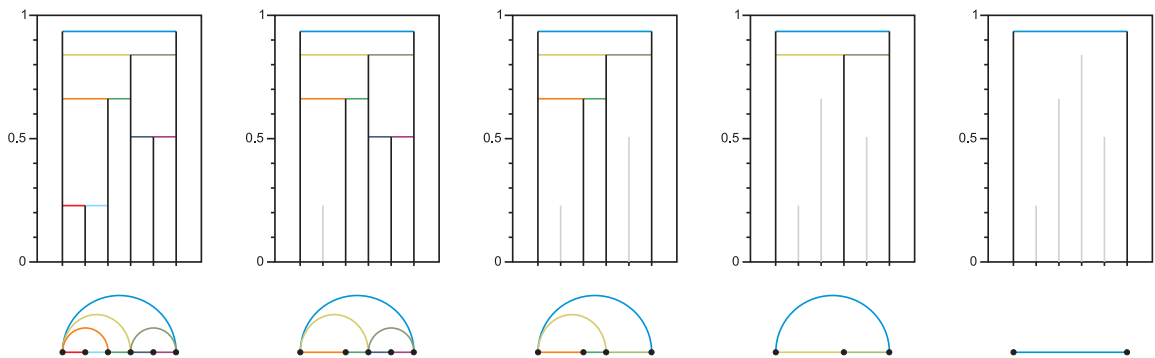}
\caption{Graphical illustration of the constructive proof of theorem \ref{kmedia}, considering a motif extracted from a periodic series of period $T=5$.} \label{ejemplo_eli}
\end{figure}
An interesting consequence of the previous result is that every time series extracted from a dynamical system has an associated HVG with a mean degree $2\leq \bar k\leq4$, where the lower bound is reached for constant series, whereas the upper bound is reached for aperiodic series (random, chaotic \cite{pre}).

\begin{theorem}
\label{p_k}
(Degree distribution associated to uncorrelated random series)\\Let $\{x_i\}$ be a bi-infinite sequence of independent and identically distributed random variables extracted from a continous probability density $f(x)$. Then, the degree distribution of its associated horizontal visibility graph is
\begin{equation}
P(k)=\frac{1}{3}\bigg(\frac{2}{3}\bigg)^{k-2},\ k=2,3,4,...
\label{pk}
\end{equation}
\end{theorem}
A lengthy constructive proof can be found in [Luque \textit{et al.} 2009]. Here we propose two alternative, shorter proofs for this theorem.
\begin{proof}Let $x$ be an arbitrary datum of the aforementioned series. The probability of its horizontal visibility being interrupted by a datum $x_r$ on its right and other datum $x_l$ on its left is, independently of $f(x)$,

\begin{equation}
\Phi_2=
\int_{-\infty}^{\infty}\int_{x}^{\infty}\int_{x}^{\infty}f(x)f(x_r)f(x_l)dx_l dx_r dx=\int_{-\infty}^{\infty} f(x)[1-F(x)]^2dx=\frac{1}{3}
\label{phi}
\end{equation}

The probability $P(k)$ of the datum seeing exactly $k$ data may be expressed as

\begin{equation}
P(k)=Q(k)\Phi_2 = \frac{1}{3} Q(k) \label{recurrence}
\end{equation}

where $Q(k)$ is the probability of $x$ seeing at least $k$ data. $Q(k)$ may be recurrently calculated as

\begin{equation}
Q(k)=Q(k-1)(1-\Phi_2)=\frac{2}{3}Q(k-1)
\end{equation}

from which the following expression can be deduced:

\begin{equation}
Q(k)= \bigg(\frac{2}{3}\bigg)^{k-2}
\end{equation}
\qed
\end{proof}
\begin{proof}(2) The same result for the distribution $P(k)$ can be deduced by means of expression $\Phi_2$ along with combinatoric arguments: if the arbitrary datum $x$ is connected to exactly $k$ data, there exist two data $x_l>x$ and $x_r>x$ that close the left and right visibility of $x$. As we have proven, this happens with a probability $\Phi_2 = 1/3$. The $k-2$ remaining data will therefore be smaller than $x$ and will be distributed in a monotonically decreasing sequence on its left and a monotonically increasing sequence on its right respectively. The number of possible distributions with $i$ data on its left and $k-2-i$ on its right is $k-2 \choose i$, where $i=0,1,2,\dots,k-2$. All these configurations can all be decomposed in $k-2$ groups of three data with the central datum being closed by the two others, therefore, all of them are equiprobable with a probability $\phi_2^{k-2}$, then

\begin{equation}
P(k) =\frac{1}{3}\sum_{n=0}^{k-2}\bigg(\frac{1}{3}\bigg)^{k-2} {k-2 \choose n} = \bigg(\frac{1}{3}\bigg)^{k-1} 2^{k-2}=
\frac{1}{3}\bigg(\frac{2}{3}\bigg)^{k-2}\label{teorico}
\end{equation}
\qed
\end{proof}
\noindent Observe that the mean degree $\bar k$ of the horizontal visibility graph associated to an uncorrelated random process is then:
\begin{equation}
\bar k =\sum kP(k)=\sum_{k=2}^{\infty}\frac{k}{3}\Bigg{(}{\frac{2}{3}}\Bigg{)}^{k-2}
= 4, \label{media_total}
\end{equation}
in good agreement with the prediction of the previous theorem for aperiodic series.

\subsection{Stochastic, chaotic and periodic processes}
Deviations from $P(k)=(\frac{1}{3})(\frac{2}{3})^{k-2}$ are, according to previous theorem, univoquely associated to series which are not generated by a purely uncorrelated process. Several possibilities arise: first, the process can still be of a stochastic nature, while some correlations can be present. As a matter of fact, in [Lacasa \& Toral, 2010] it has been shown that such kind of correlated stochastic series map into HVGs with a degree distribution which is also exponentially decaying, albeit with a larger slope than equation \ref{pk}. A second situation involves a deterministic process. Two opposite possibilities arise: the process can be either regular, what yields a periodic series of a given period $T$, or chaotic, what yields an aperiodic series. Periodic series have an associated HVG with a degree distribution formed by a finite number of peaks,
these peaks being related to the series period, what is reminiscent of the discrete Fourier spectrum of a periodic series [Lacasa \textit{et al.} 2008, Luque \textit{et al.} 2009].
The reason is straightforward: a periodic series maps into an HVG which, by construction, is a repetition of a root motif.
The second possibility has been addressed in [Luque \textit{et al.} 2009, Lacasa \& Toral, 2010], the conclusion being that chaotic processes have an HVG whose degree distribution has an exponential tail with smaller slope than equation \ref{pk}, and evidences a net deviation from the exponential shape for small values of the degree, this deviation being associated to short-range memory effects. Last, an interesting situation takes place when a given regular process (periodic series) is polluted with a given amount of noise.

\begin{figure}[h]
\centering
\includegraphics[width=0.7\textwidth]{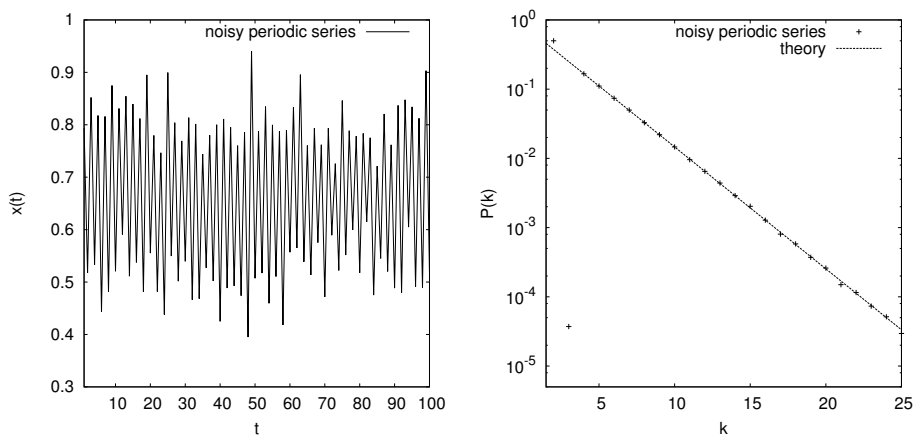}
\caption{\textit{Left}: Periodic series of $2^{20}$ data generated through the logistic map $x_{n+1}=\mu x_n(1-x_n)$ for $\mu=3.2$ (where the map shows periodic behavior with period $2$)
polluted with extrinsic white gaussian noise extracted from a Gaussian distribution $N(0,0.05)$. \textit{Right}: Dots represent the degree distribution of the associated HVG, whereas
the straight line is equation \ref{teoriaperi} (the plot is in semi-log). Note that $P(2)=1/2$, also as theory predicts, and that $P(3)$ is not exactly zero due to boundary effects in the time series.
The algorithm efficiently detects both signals and therefore easily distinguishes extrinsic noise.} \label{ejemplo}
\end{figure}

Indeed, if we superpose a small amount of noise to a periodic series (a so-called \textit{extrinsic} noise),
while the degree of the nodes with associated small values will remain rather similar, the nodes associated to higher values will eventually increase their visibility and hence
reach larger degrees. Accordingly, the delta-like structure of the degree distribution (associated with the periodic component of the series) will be perturbed, and an exponential
tail will arise due to the presence of such noise [Luque \textit{et al.} 2009, Lacasa \& Toral, 2010]. Can the algorithm characterize such kind of series? The answer is positive, since the degree distribution can be analytically calculated as it follows:\\
Consider for simplicity a period-2 time series polluted with white noise (see the left part of figure \ref{ejemplo} for a graphical illustration).
The HVG is formed by two kind of nodes: those associated to high data with values $(x_1,x_3,x_5,...)$ in the figure and those associated to data with small values
$(x_2,x_4,x_6,...)$. These latter nodes will have, by construction, degree $k=2$. On the other hand, the subgraph formed by the odd nodes associated to data $(x_1,x_3,x_5,...)$ will essentially reduce to the one
associated to an uncorrelated series, \textit{i.e.} its degree distribution will follow equation \ref{pk}. Now, considering the whole graph, the resulting degree distribution will
be such that
\begin{eqnarray}
&&P(2)=1/2,\nonumber\\
&&P(3)=0,\nonumber\\
&&P(k+2)=\frac{1}{3}\bigg(\frac{2}{3}\bigg)^{k-2}, \ k\geq2, \nonumber\\
&&\text{or}\  P(k)= \frac{1}{4}\bigg(\frac{2}{3}\bigg)^{k-3}, \ k\geq4,\label{teoriaperi}
\end{eqnarray}
that is to say, introducing a small amount of extrinsic uncorrelated noise in a periodic signal introduces an exponential tail in the HVG's degree distribution with the same slope as the one associated to a purely uncorrelated process. The mean degree $\bar k$ reads
$$\bar k=\sum_{k=2}^\infty kP(k)=4,$$
which, according to equation \ref{kmedia}, suggests aperiodicity, as expected.
In the left part of figure \ref{ejemplo} we plot in semi-log the degree distribution of a periodic-2 series of $2^{20}$ data polluted with an extrinsic white Gaussian noise
extracted from a Gaussian distribution $N(0,0.05)$. Numerical results confirm the validity of equation \ref{teoriaperi}. Note that
this methodology can be extended to every integrable deterministic system, and therefore we conclude that extrinsic noise in a mixed time series is, in principle, well captured by the algorithm. Based on the previous theorems, in the next section we introduce a method to calculate the hidden periodicity in a noisy periodic signal.

\section{A graph-theoretical noise filter}
\subsection{Definition and examples}
Let $S=\{x_i\}_{i=1,...,n}$ be a periodic series of period $T$ (where $n>>T$) polluted by a certain amount of extrinsic noise (without loss of generality, suppose a white noise extracted from a uniform distribution $U[-0.5,0.5]$), and define the filter $f$ as a real valued scalar such that $f\in[\min x_i,\max x_i]$. The so called filtered Horizontal Visibility Graph (f-HVG) associated to $S$ is constructed as it follows:\\
(i) each datum $x_i$ in the time series is mapped to a node $i$ in the f-HVG,
(ii) two nodes $i$ and $j$ are connected in the fHVG if the associated data fulfill
\begin{equation}
x_i,x_j >x_n + f, \ \forall \ n \ \left | \ i < n < j\right. . \label{criteriofiltro}
\end{equation}
The procedure of filtering the noise from a noisy periodic signal goes as follows: one generates the f-HVG associated to $S$ for increasing values of $f$, and in each case proceeds to calculate the mean degree $\bar k$. For the proper interval $f_{min}<f<f_{max}$, the f-HVG of the noisy periodic series $S$ will be equivalent to the noise free HVG of the pure (periodic) signal, which has a well defined mean degree as a function of the series period. In this interval, the mean degree will therefore remain constant, and from equation \ref{kmedia} the period can be inferred.\\

As an example, we have artificially generated a noisy periodic series of hidden period $T=2$, that we plot in the left panel of figure \ref{T2}. The results of the graph filtering technique are shown in the middle panel of this figure, where we plot the values of $\bar k$ as a function of $f$. Notice that the graph filtering yields a net decreases of the mean degree, which has an initial value of $4$ (as expected for the HVG ($f=0$) of an aperiodic series such as a noisy periodic signal) and an asymptotic value of $2$ (lower bound of the mean degree).
The plateau is clearly found at $\bar k=3$, which according to equation \ref{kmedia} yields a period
$$T=\bigg(2-\frac{\bar k}{2}\bigg)^{-1}=2,$$ as expected. For comparison, the autocorrelation function $ACF(\tau)$ of the series is also calculated, according to the following definition
$$ACF(\tau)=<x(t)\cdot x(t-\tau)>_t,$$
such that $ACF$ is not bounded in $[-1,1]$ since it is not normalized. This expression has a periodic shape of period $T$ when the series is itself periodic with period $T$, whereas aperiodic structures yield an autocorrelation function that lacks any structure. In the right panel of figure \ref{T2} we plot the values of the autocorrelation, showing a period-2 structure as expected.\\
At this point we conclude that the noise filtering is yet another feature of standard time series analysis that can be recovered in the visibility theory. An example with a noisy periodic series of period $T=5$ is plotted in figure \ref{T5}.
\begin{figure}[h]
\centering
\includegraphics[width=1.0\textwidth]{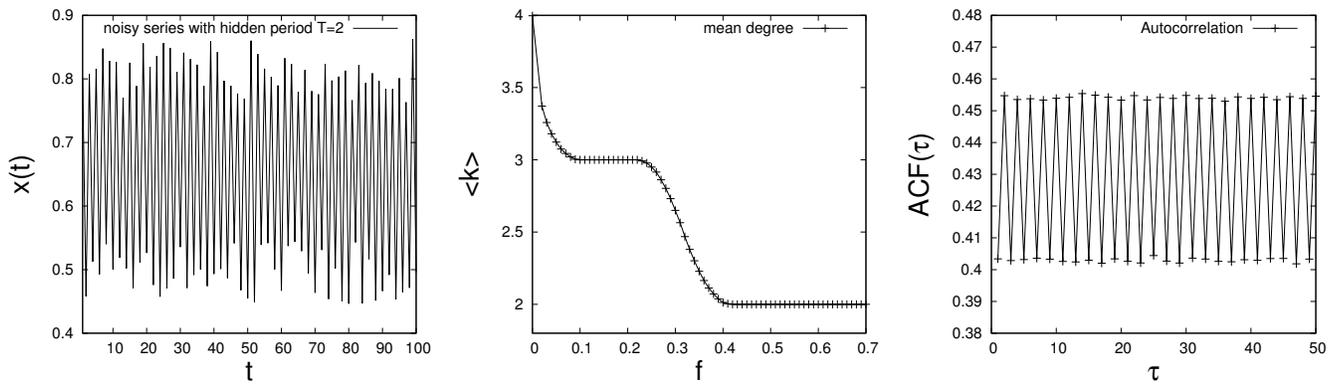}
\caption{\textit{Left}: Periodic series of period $T=2$ polluted with extrinsic noise extracted from a uniform distribution $U[-0.5,0.5]$ of amplitude $0.1$. \textit{Middle}: Values of the HVG's mean degree $\bar k$ as a function of the amplitude of the graph theoretical filter. The first plateau is found for $\bar k=3$, which renders a hidden period $T=(2-\bar k/2)^{-1}=2$. The second plateau corresponding to $\bar k=2$ is found when the filter is large enough to screen each datum with its first neighbors, such that the mean degree reaches its lowest bound. \textit{Right}: Autocorrelation function of the noisy periodic series, which is itsef an almost periodic series with period $T=2$, as it should.} \label{T2}
\end{figure}

\begin{figure}[h]
\centering
\includegraphics[width=1.0\textwidth]{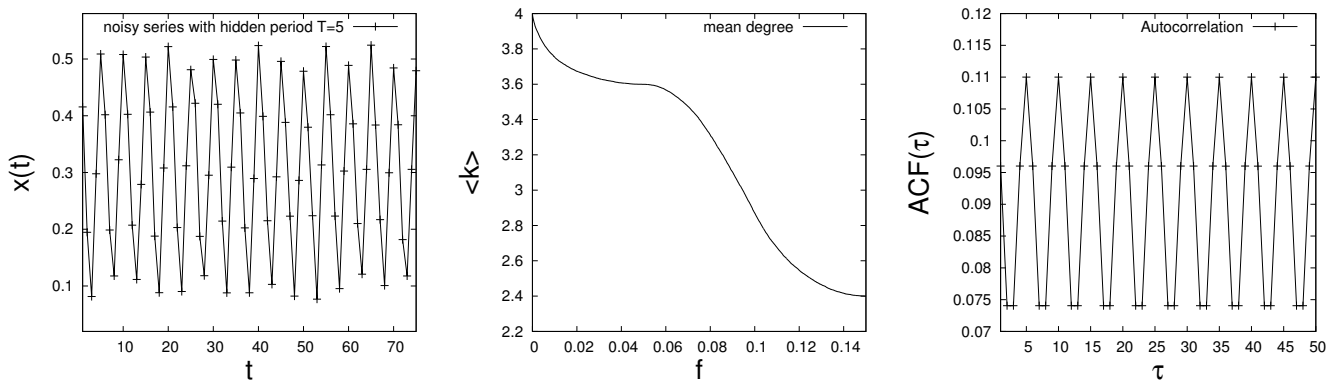}
\caption{\textit{Left}: Periodic series of period $T=5$ polluted with extrinsic noise extracted from a uniform distribution $U[-0.5,0.5]$ of amplitude $0.05$. \textit{Middle}: Values of the HVG's mean degree $\bar k$ as a function of the amplitude of the graph theoretical filter. The first plateau is found for $\bar k=3.6$, which renders a hidden period $T=(2-\bar k/2)^{-1}=5$.  \textit{Right}: Autocorrelation function of the noisy periodic series, which is itsef an almost periodic series with period $T=5$, as it should.} \label{T5}
\end{figure}

\subsection{Noisy periodic versus chaotic}
\begin{figure}[h]
\centering
\includegraphics[width=1.0\textwidth]{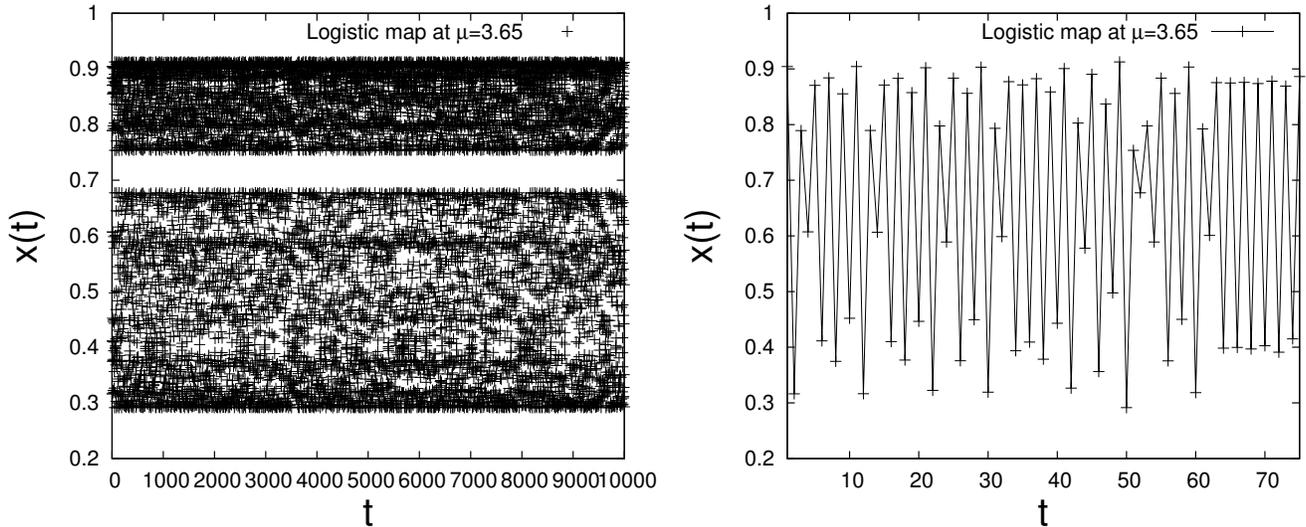}
\caption{\textit{Left}: Series extracted from the Logistic map at $\mu=3.65$, where the map is chaotic and the attractor is partitioned in two disconnected chaotic bands. \textit{Right}: Same plot as the left panel, for the first $75$ values of the series.} \label{T2_caos}
\end{figure}
\begin{figure}[h]
\centering
\includegraphics[width=1.0\textwidth]{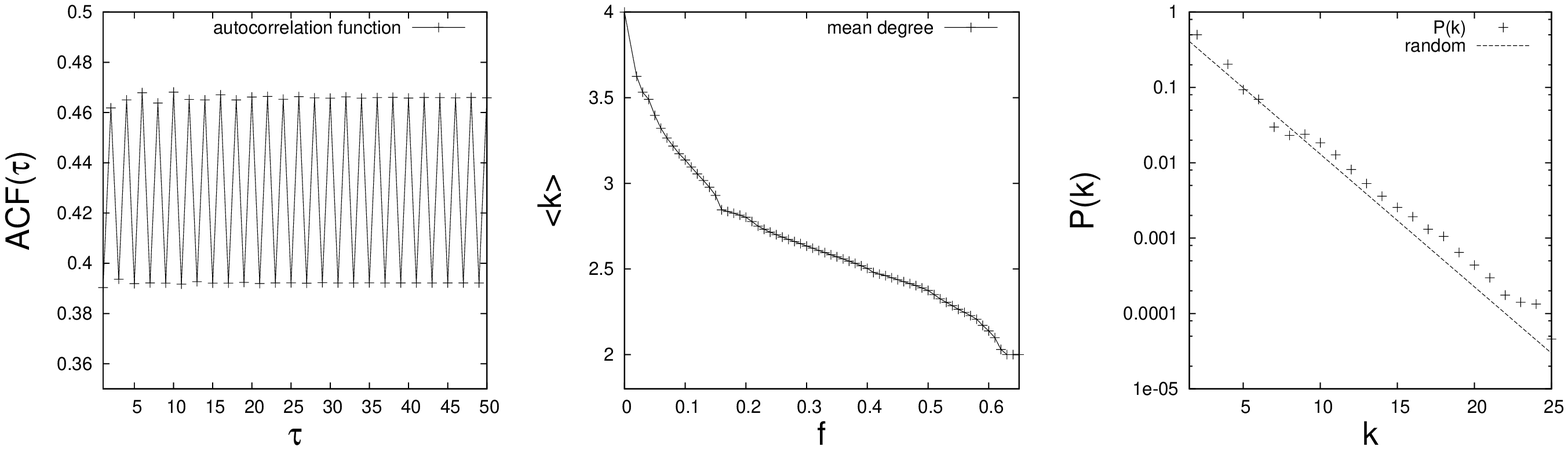}
\caption{\textit{Left}: Autocorrelation function of the chaotic series plotted in figure \ref{T2_caos}, which suggests that the series has a periodic component. This misleading result is consequence of the orbit within the attractor, more specifically to the alternated visit to the disconnected bands in the chaotic attractor. \textit{Middle}: Values of the HVG's mean degree $\bar k$ associated to the same chaotic series, as a function of the amplitude of the graph theoretical filter. No plateau is found, what suggests that the series lacks any periodic structure, as it should.\textit{Right}: Degree distribution $P(k)$ of the HVG associated to the chaotic series, in semi-log scale. $P(2)=1/2$, what is reminiscent of the attractor structure and the order of visits to chaotic bands (half of the nodes correspond to data located in the bottom chaotic band, that by construction has degree $k=2$). The tail of the distribution is exponential with a slope that deviates from the distribution associated to a purely uncorrelated process, what is an indication of a chaotic process, according to a previous study on HVGs \cite{pre2}.} \label{T2_caos_todas}
\end{figure}
The autocorrelation function is an extremely useful tool to unveil periodic structures in noisy data, in this sense the aforementioned filter is not meant to be used \textit{instead} an autocorrelation analysis, but rather as a complementary study. As a matter of fact, in specific situations it may happen that an autocorrelation analysis may provide misleading results. This is for instance the case of chaotic maps with disconnected attractors. Consider the well known Logistic map $$x_{t+1}=\mu x_t(1-x_t),$$ with $\mu\in[0,4]$ and $x\in[0,1]$. This map generates periodic series for $\mu<\mu_\infty=3.569...$, while for $\mu>\mu_\infty$ the map generates chaotic (deterministic and aperiodic) series (besides regions where the orbit turns regular again, called islands of stability [Peitgen \textit{et al.} 1992]). In the chaotic region, the chaotic attractor is the whole interval $[0,1]$ only for $\mu=4$. Concretely, for $\mu\in[3.6,3.67]$ the attractor is partitioned in two disconnected chaotic bands, and the chaotic orbit makes an alternating journey between both bands. In figure \ref{T2_caos} we have plotted a time series of $2^{18}$ data generated through the Logistic map at $\mu=3.65$. Note that the map is ergodic, but the attractor is not the whole interval, as there is a gap between both chaotic bands.
In this situation, the chaotic series is by definition not periodic, however, an autocorrelation function analysis indeed suggests the presence of periodicity (see the left panel of figure \ref{T2_caos_todas}), what is reminiscent of the disconnected two-band structure of the attractor. Interestingly enough, applying the aforementioned noise filter technique, at odds with the autocorrelation function, the results suggests that the method does not find any periodic structure, as it should (middle panel of figure \ref{T2_caos_todas}). Furthermore, information of both the phase space structure and the chaotic nature of the map becomes accesible from an analysis of the HVGs degree distribution (plotted in semi-log scale in the right panel of figure \ref{T2_caos_todas}): first, we find $P(2)=1/2$, that indicates that half of the data are located in the bottom chaotic band, in agreement with the alternating nature of the chaotic orbit. This is reminiscent of the misleading result obtained from the autocorrelation function. Second, the tail of the degree distribution is exponential, with an asymptotic slope smaller than the one obtained rigorously (theorem \ref{p_k} and \cite{pre}) for a purely uncorrelated process. This is, according to a recent study on HVGs [Lacasa \& Toral, 2010], characteristic of an underlying chaotic process.

\section{Discussion}
In this work we have outlined some properties of the HVGs associated to time series extracted from dynamical systems, and have accordingly proposed a method to detect periodicity in signals polluted with noise. The results suggest that the HVG correctly inherits the hidden periodicity of noisy signals, and can be retrieved by making use of the aforementioned filter in situations where the noise level is not very large (indeed, the maximum power of the noise is of order $O(\Delta x^2)$, where $\Delta x= \min\{|x_i-x_j|\}_{i,j}$).\\
Also, we have found that specific pathological cases where a
classical time series analysis yields misleading results, such as
in chaotic (aperiodic) series generated from chaotic maps with a
disconnected attractor, can be efficiently analyzed within this
network-based tool. This approach is radically different from
traditional methods for time series analysis since it is based on
graph theoretical properties, and therefore can be used as a
complementary tool in practical situations. We have deliberately
not tackled the task of systematically comparing the goodness of
such graph theoretical method with other standard tools of time
series analysis, and have restricted this comparison to checking
that an autocorrelation function analysis provides equivalent
results for the cases addressed here. In this sense, we must
emphasize that the goal of this work is to push forward the state
of the art in the visibility theory providing mathematically sound
properties, rather than putting the practical usefulness of both
methods in direct competence.\\\\

\noindent \textbf{\emph{Acknowledgments}} \noindent The authors
acknowledge financial support by the MEC (Spain) through project
FIS2009-13690 and Comunidad de Madrid through research program
MODELICO (S2009/ESP-1691).

\end{document}